\documentclass[12pt]{article}
\usepackage{verbatim}
\usepackage{amsfonts}
\usepackage{graphics}
\usepackage{amsmath}
\usepackage{times}
\usepackage{appendix}
\usepackage{color}
\usepackage{soul}
\usepackage{mathtools}
\usepackage{enumerate}
\usepackage{fancyhdr,latexsym,amsmath,amsfonts,amssymb,amsbsy,amsthm,url}
\usepackage[margin=0.5in,footskip=0.25in]{geometry}
\usepackage{graphics,graphicx,epsfig}
\usepackage{bbm}
\usepackage{breqn}
\usepackage{caption,subcaption,float} 
\usepackage{pdflscape}
\usepackage{hyperref}
\usepackage{tikz}
\usepackage[ruled,vlined]{algorithm2e}

\newtheorem{theorem}{Theorem}
\newtheorem{lemma}[theorem]{Lemma}

\newtheorem{corollary}[theorem]{Corollary}

\usepackage[english]{babel}
\makeatletter

\renewcommand{\section}{
	\@startsection
	{section}
	{1}
	{0pt}
	{1.1\baselineskip}
	{0.2\baselineskip}
	{\sc \centering}
}

\renewcommand{\subsection}{
	\@startsection
	{subsection}
	{1}
	{0pt}
	{1.1\baselineskip}
	{0.2\baselineskip}
	{\sc \centering}
}

\renewcommand{\subsubsection}{
	\@startsection
	{subsubsection}
	{1}
	{0pt}
	{1.1\baselineskip}
	{0.2\baselineskip}
	{\sc \centering}
}

\makeatother

\usepackage[flushleft]{threeparttable}
\usepackage{rotating,booktabs,multirow}
\usepackage{colortbl}
\usepackage{makecell,cellspace,caption}

\begin{document}

\title{\large\sc Can Limited Liability Increase Stability for Banks: A Dynamic Portfolio Approach}
\normalsize
\author{
\sc{Deb Narayan Barik} \thanks{Department of Mathematics, Indian Institute of Technology Guwahati, Guwahati-781039, India, e-mail: d.narayan@iitg.ac.in}
\and 
\sc{Siddhartha P. Chakrabarty} \thanks{Department of Mathematics, Indian Institute of Technology Guwahati, Guwahati-781039, India, e-mail: pratim@iitg.ac.in, Phone: +91-361-2582606}}

\date{}
\maketitle
\begin{abstract}
We present a novel approach for the bank's decision problem, incorporating Limited Liability in the objective function. Accordingly, we consider continuous time models, with and without Limited Liability. We compare the solutions of these two models to demonstrate the effect of inclusion of Limited Liability. To solve the problem with the objective function incorporating Limited Liability, we approximate the payoff function to another set of functions for which we have closed-form solutions. Then, we show that the solution with Limited Liability incorporates less risky assets, while simultaneously increasing the resilience of the bank. After that, we use the metric of $Distance~to~Default$, from the KMV Model, to analyze the bank's resiliency, by considering that the interest rate follows the Vasicek model. Finally, we illustrate the results obtained with a numerical example.
\end{abstract}

\section{Introduction}
\label{Sec_Four_Introduction}

The continuous-time portfolio selection problem, was first established by Merton \cite{Merton1973,Merton74}, wherein the model was set up to optimize the portfolio selection problem, in continuous time, from the perspective of an investor. The problem formulation sought to determine the portfolio, that needs to be held between the current time and the finite final time $T$, via maximization of the expected terminal utility of the wealth level. The portfolio under this model framework is assumed to be comprised of a risk-free asset and a finite number of risky assets, which are accessible for the purpose of investment. The optimization problem (to determine the optimal portfolio) reduces to a stochastic control problem. One drawback of this seminal model was that it assumed the risk-free rate of interest to be a constant. This shortcoming (pertaining to the constant risk-free interest rate) was addressed by Vasicek \cite{Vasicek1977}, and Ho and Lee \cite{Ho1986}. Both these articles consider the interest rate to be following a stochastic process, that is both free of arbitrage and matches the current yield curve (implied by the pricing of interest rate derivatives). Vasicek \cite{Vasicek1977} had considered a mean-reverting process for the interest rate and provides an analytic solution for the interest rate derivatives. This, in turn, has led to the determination of the bond pricing formula. The price of the bond can be correlated with the risk-premium in case of interest rate risk. In a more recent work \cite{Korn2002}, the authors have revisited the portfolio optimization problem, this time for a portfolio comprising of a savings account and some risky bonds.

In the banking structure, there exists the concept of Limited Liability, which protects the owners of the bank (or by extension, any Limited Liability firm) form personal liability, in the event o bankruptcy or insolvency, that is, in such scenarios, the depositors cannot lay any claim on the owner's personal assets, for their deposit exposure. The interested reader may refer to \cite{Carney1998}, for a historical discourse on the concept of Limited Liability. The protection accorded by Limited Liability (naturally) causes the problem of moral hazard \cite{Sinn2001, Biais2010}. Then the bank's decision problem with Limited Liability, in conjunction with capital requirements mandated in Basel III, has been studied in \cite{Acosta2020}. Subsequently, in \cite{Barik2022}, the authors have demonstrated that inclusion of Limited Liability, in the decision model (in discrete time) can lead to decreased leveraged risk. With Limited Liability protection, the required amount of risk is less, as compared to the case without Limited Liability protection, when it comes to attaining a target.

The capital requirements as mandated in Basel I and Basel II \cite{Basel1,Basel2}, were restricted to risk-based capital measures. However, the 2008 Global Financial Crisis (GFC) led the regulators to introduce capital requirements to disincentivize banks and financial institutions from taking on unacceptable levels of leverage. Accordingly, Basel III \cite{Basel3Leverage} saw the introduction of Leverage Ratio (LR), a non-risk-based risk metric, to complement the existent risk-based metrics. The definition of LR is given by:
\[\text{LR}:=\frac{\text{Tier 1 capital}}{\text{Total exposure measure}} ,\] 
with this ratio being a counter-cyclical capital measure. Counter-cyclical measures are reportedly more effective in preempting and reducing systematic risk, as well as credit bubbles. Blum \cite{Blum2008} has demonstrated that the risk-based capital requirement, in conjunction with the LR, is a better criterion for capital requirement. D'Hulster \cite{D09} has discussed about various aspects of the LR, in the context of bank's leverage. Hildebrand \cite{Hildebrand2008} has claimed that implementing risk-based capital requirements with LR lowers the leverage of the bank, which eventually decreases the chance of default of the bank. For the purpose of this work we have considered the loan portfolio of a bank which fulfills both the risk-based capital requirement and the LR criterion.

The KMV model \cite{Valavskova2014} is a credit risk model that estimates the probability of default in the framework of Merton's structural credit risk model, which treats the firm's equity as a call option on its assets. Default occurs when the market value of assets falls below the default point. This model use a metric, namely, Distance to Default ($DD$) to measure the company's probability of default. $DD$ measures how close a company is to defaulting on its debt, with a higher value of $DD$ implying a lower risk of default.
 
\textit{The focus of the work} is to study the effect of inclusion of Limited Liability in model, for a continuous time setup. Accordingly, we construct two models, one with and another without Limited Liability, involving the maximization of expected utility function of terminal wealth. After that we study the stability of the bank against these solutions using the $DD$ from the KMV model. The outline of the paper is as follows. We present a detailed description of the models in Section \ref{Sec_Four_Formulation}. Then, in Section \ref{Sec_Four_Analysis} we derive all the theoretical results, which are then illustrated by an example in Section \ref{Sec_Four_Example}. Finally, Section \ref{Sec_Four_Conclusion} summarizes the main takeaways and concluding remarks on the work reported in the paper.

\section{Formulation of the Model}
\label{Sec_Four_Formulation}

In this Section, we briefly discuss the formulation of the Hamilton-Jacobi-Bellman (HJB) equation in the context of the problem under consideration. Accordingly, two optimization problems are formulated, one with Limited Liability and the other without Limited Liability. As stated previously, we consider the DD metric in the paradigm of the KMV model used to study the resilience of the firm. Also, the interest rate risk is taken into consideration for this model. Accordingly, we have considered the interest rate as having followed the Vasicek model, an all of which have been described in this Section.
 
\subsection{HJB Equation Formulation of an Optimization Process}

In this Subsection, we discuss briefly how the HJB equation is formulated for an optimization process. Let $\left\{(W(t)\right\}_{t \in [0, \infty)}$ be a $m$-dimensional Brownian motion. Further, let $Y(t)$ be a $n$-dimensional state process, whose evolution is given by the stochastic differential equation (SDE) as follows: 
\begin{equation}
\label{Eq_4.1}
dY(t)=\Lambda(t,Y(t),u(t))dt+\Sigma(t,Y(t),u(t))dW(t),
\end{equation}
with an initial condition of $Y(t_{0})=y_{0}$. Here $u(t)$ is a $d$-dimensional control process, $\Lambda(\cdot)$ is the drift, and $\Sigma(\cdot)$ is the instantaneous standard deviation. Let $[t_{0},t_{1}]$ with $0 \leq t_{0} < t_{1} < \infty$ be the pertinent time window. Our goal is to determine an admissible control $u(\cdot)$ such that for each initial value $(t_{0},y_{0})$, the utility functional, 
\begin{equation}
\label{Eq_4.2}
J(t_{0},y_{0};u):=\mathbb{E}^{t_{0},y_{0}}\left(\int\limits_{t_{0}}^{T} L(s, Y^{u}(s),u(s))dt+\Psi(T,Y^{u}(T))\right),
\end{equation}
is to be maximized. Here we want to maximize the functional up to time $T$, a fixed time horizon. Therefore, we want to solve the problem:
\[\sup_{u \in \mathcal{A}(t_{0},y_{0})}J(t_{0},y_{0};u).\]
Accordingly, the value function is defined as,
\[V(t,y):= \sup_{u \in \mathcal{A}(t,y)} J(t,y;u),~(t,y) \in Q.\]
Here $Q=[t_{0},t_{1}) \times O$ where $O \in \mathbb{R}^{n}$ is an open set.
For each function $G \in C^{1,2}(Q)$ and $\left(t,y\right) \in Q $, ($v \in U$ where $U$ is a subset of $\mathbb{R}^{d}$), we consider the following differential operator:
\begin{equation}
\label{Eq_4.3}
A^{v}G(t,y):=G_{t}(t,y)+0.5 \sum\limits_{i,j=1}^{n} \Sigma_{ij}^{*}(t,y,v) \cdot G_{y_{i} y_{j}}(t,y)+\sum\limits_{i=1}^{n} \Lambda_{i}(t,y,v) \cdot G_{y_i}(t,y),
\end{equation}
where $\displaystyle{\Sigma^{*}:=\Sigma\Sigma^{\top}}$. From the verification theorem \cite{Pham2009}, if the problem, 
\begin{equation}
\label{Eq_4.4}
\sup_{v\in U} \left[A^{v}G(t,y)+L(t,y,v)\right]=0,~(t,y)\in Q~\text{and}~
G(t,y)=\Psi(t, y),~(t, y) \in \partial Q.
\end{equation}
has a solution, then the following holds:
\begin{enumerate}[(A)]
\item $G(t,y) \geq J(t,y;u)$ for all $(t,y) \in Q$ and $u(\cdot) \in \mathcal{A}(t,y)$.
\item If for $(t,y) \in Q$ there exists a control $u^{*}(\cdot) \in \mathcal{A}(t,y)$ with,
\begin{equation}
\label{Eq_4.6}
u^{*}(s) \in \arg\max_{v \in U} \left[A^{v} G(s,Y^{*}(s))+L(s,Y^{*}(s),v) \right],
\end{equation}
for all $s \in [t,\tau]$ (where $Y^{*}$ is the solution of the controlled SDE corresponding to $u^{*}(\cdot)$), then we have,
\begin{equation}
\label{Eq_4.7}
G(t,y)=V(t,y)=J(t,y;u^{*}),
\end{equation}
\textit{i.e.,} $u^{*}(\cdot)$ is an optimal control and $G$ coincides with the value function.
\end{enumerate}

\subsection{Problem Description}

We construct the problem with a Money Market Account (MMA) and a risky bond. Let $\pi(t)$ denote the amount invested in the risky bond. Our goal is to determine $\pi(t)$ so that the terminal wealth is maximized. As mentioned earlier, the interest rate $r(t)$ follows the Vasicek model (discussed elaborately in \cite{Puhle2008}) given by,
\[dr(t)=a(t)dt+bdW(t),\]
where $a(t)$ is given by, 
\[a(t)=\theta(t)-\alpha r(t),~\alpha > 0.\]
Here $\alpha$ and $\displaystyle{\frac{\theta(t)}{\alpha}}$ are the mean reversal speed and level, respectively. Further, $b$ is the instantaneous standard deviation. Accordingly, we begin with a portfolio where the bank can invest in a MMA and a bond (zero coupon) with maturity $T_{1}>T$, with $T$ being the time horizon in which we are seeking to optimize our wealth level. The bond price using risk premium is presented in \cite{Korn2002,Munk2004}. The dynamics of price of the risky bond is being presented by, 
\[dP(t,T_{1})=P(t,T_{1})\left[\mu(t)dt+\sigma(t)dW(t)\right].\]
We take $\mu(t)=r(t)+\sigma(t)\zeta (t)$ ($\zeta(t)$ being the risk premium), as motivated by \cite{Korn2002}. Also, we take \\ $\displaystyle{\sigma(t)=\frac{b}{\alpha}\left(1-\exp{(-\alpha(T_{1}-t))} \right)}$. Now the total wealth level follows the SDE, 
\begin{eqnarray}
dX(t) &=& X(t)\left[(\pi(t)\mu(t)+(1-\pi(t))r(t))dt+\pi(t)\sigma(t) dW(t)\right], \nonumber\\
&=& X(t)\left[(\pi(t)\zeta(t)\sigma(t)+r(t)) dt+\pi(t)\sigma(t) dW(t)\right].
\label{Eq_4.8}
\end{eqnarray}

Now our objective is to maximize the utility of the terminal wealth $U(X(T))$. Consequently, the optimization problem becomes, 
\[\max_{\pi(\cdot) \in \mathcal{A}(0,x_{0})} \mathbb{E}\left[U(X(T))\right],\]
with the initial condition $X(0)=x_{0}=1$ (normalizing), and $\mathcal{A}(0,x_{0})$ is the set of admissible controls corresponding to the initial condition $t=t_{0}$ and $x=x_{0}$.
Now, comparing with equation (\ref{Eq_4.1}), we get,
\begin{eqnarray}
\label{Eq_4.9}
Y(t) &=& (X(t),r(t))^{\top}, \nonumber \\
\Lambda(t,x,r,\pi) &=& (x(\pi \zeta \sigma+r),a)^{\top}, \nonumber \\
\Sigma(t,x,r,\pi) &=& (x \pi \sigma, b)^{\top}, \nonumber \\
\Sigma^{*}(t,x,r,\pi)&=& 
\begin{pmatrix}
x^{2}\pi^{2} \sigma^{2} & bx \pi \sigma \\
bx \pi \sigma & b^{2}
\end{pmatrix}, 
\nonumber\\
A^{\pi} G(t,x,r)&=& G_{t}+0.5\left(x^{2} \pi^{2} \sigma^{2} G_{xx}
+2x\pi b \sigma G_{x r}+b^{2} G_{rr}\right), \nonumber \\
&+& x(\pi\zeta \sigma+r)G_{x}+aG_{r}.
\end{eqnarray}
Hence, in order to get the desired optimal control we need to solve the following HJB equation:
\begin{equation}
\label{Eq_4.10}
\sup_{|\pi| \leq \delta} A^{\pi} G(t,x,r)=0;~G(T,x,r)=U(X(T)), 
\end{equation}
with $\delta$ being the upper-bound of the investment in the risky asset. The allocation is made by bounding the Expected Loss ($EL$) by $\theta$, for the entire portfolio. For the MMA, the probability of default is zero, and it has no contribution in $EL$. So the $EL$ of the portfolio is,
\[EL=\pi(t) p \lambda \leq \theta \implies \pi(t) \leq \frac{\theta}{p\lambda}=\delta~(\text{bound on risky asset}).\]
In the above equation, $p$ and $\lambda$ denote the probability of default and loss given default, respectively. In case we take the utility function of the form $U(x)=x^{\gamma},~\gamma \in \mathbf{R}$ then we can find a closed feedback form for $\pi(t)$ as a function of investment-instrument's characteristic. As already noted, the goal is to study the effect of inclusion of Limited Liability in the model. To make a comparison, we need a closed form solution (for the investment portfolio), and therefore, we approximate the pay-off function with some function of the form $x^{\gamma}$, where $\gamma \in \mathbf{R}$ in $L^{2}$-norm. For the case without Limited Liability, the pay-off function, $X$ is directly approximated by $X^{1}$. Therefore, the utility function becomes,
\[U_{1}(X)=X,\]
which represents the objective function without Limited Liability. In this case, we know that the solution exists when the amount invested in the risky asset has a cap, the amount invested in the risky asset plateaus the upper-bound given by the constraint. In this context, Limited Liability means that the pay-off function cannot fall below a certain level, because if that happens, then the bank faces bankruptcy. In this study, we denote this bankruptcy level by $F$, that is, the payoff function in this case becomes $\max\left(F,X\right)$. In this case, the boundary function is a Lipschitz continuous function. We can approximate it with a $C^{2}$ function  in the above-mentioned functions. Later we show that the $\gamma$ for the best approximated function lies in $(0,1)$. In this case, the optimization model is the same as the model mentioned in \cite{Korn2002}. Then we find the optimal amount invested in the risky asset. Note that the expected terminal wealth is continuously dependent on the wealth level, and in turn, the wealth level is also continuously dependent on the amount invested in the risky asset. Hence, expected terminal wealth is continuously dependent on the amount invested in the risky asset. Therefore we can compare the objective functions with their approximated versions. Now, for this approximation, we need a cap on $X$, which is denoted by $B$, that is, the maximum wealth of the portfolio in one year ($T=1$) can be $B$. Finally, we get, 
\[U_{2}(X)=X^{\gamma^{*}}\] 
which is the best approximation of $\max\left(F,X\right)$ in the above mentioned set. Then, in order to discuss the final effect we use the $DD$ concept from KMV Model. The DD defined as:
\[DD(\mu(t),D,\sigma(t),t):=\frac{\ln\left(\frac{V_{A}}{D}\right)+\left(\mu-\frac{1}{2} \sigma(t)^{2}\right) t}{\sigma(t)\sqrt{t}}.\]
Here, $V_{A}$ denotes the initial wealth and $D$ is the face value of the debt. When this $DD$ decreases, the chance of default increases and hence the stability decreases. The interested reader may refer to \cite{Crosbie2019} for a detailed description of this phenomenon.

\section{Analysis}
\label{Sec_Four_Analysis}

In this section we present the theoretical results of this study. To begin with, we study the relationship between the DD and the risk of the bank portfolio.

\begin{theorem}
$DD$ reduces when the holding in riskier security increases.
\end{theorem}
\begin{proof}
Differentiating $DD$ with respect to $\sigma$, we get,
\[\frac{\partial}{\partial \sigma}(DD)= \frac{-rt\sqrt{t}-\frac{\sigma{2}}{2}t\sqrt{t}}{(\sigma \sqrt{t})^{2}}.\]
Here $\mu(t)=r(t)+\sigma(t)\zeta(t)$ and $\zeta(t)$ represents the risk premium. From the above equation, we can see that $\displaystyle{\frac{\partial}{\partial \sigma} (DD)<0}$, and hence it is clear that $DD$ decreases as the $\sigma(t)$ increases.
\end{proof}

An interesting question pertains to how does the final portfolio vary with the $\gamma$ in the utility function $U(x)=x^{\gamma}$. In the following theorem, we have studied the relation between this $\gamma$ and the final portfolio.

\begin{theorem}
When $\gamma$ in $x^{\gamma}$ increases, then the holding in the risky asset increases. 
\end{theorem}
\begin{proof}
Proceeding on the lines of \cite{Korn2002}, we get the optimal portfolio as,
\[\pi^{*}(t)=\max\left(\frac{1}{1-\gamma} \cdot \frac{\zeta(t)+b \beta(t)}{\sigma(t)},\delta \right),\]
where,
$\displaystyle{\beta(t)=\frac{\gamma}{\alpha}\left(1-\exp (\alpha(t-T))\right)}$. 
If the optimal holding plateaus at the maximum amount, then there is nothing to prove. Accordingly, let us consider the case when,
\[\pi^{*}(t)=\frac{1}{1-\gamma} \cdot \frac{\zeta(t)+b \beta(t)}{\sigma(t)}.\]
Now substituting this $\beta(t)$ in the equation, we get,
\[\pi^*(t)= \frac{1}{1-\gamma} \cdot \frac{\zeta(t)+b \frac{\gamma}{\alpha}(1-\exp (\alpha(t-T)))}{\sigma(t)}\]
Let $\displaystyle{k(t):= \frac{b}{\alpha}(1-\exp (\alpha(t-T)))}$. Then substituting this $k(t)$ in the preceding equation, we get,
\[\pi^{*}(t)=\frac{1}{1-\gamma} \cdot \frac{\zeta(t)+\gamma k(t)}{\sigma(t)}.\]
Finally, differentiating this w.r.t $\gamma$, we get,
\begin{eqnarray*}
\frac{\partial \pi^{*}(t)}{\partial t} &=& \frac{(1-\gamma)\sigma(t)k(t)+ \sigma(t)(\zeta(t)+\gamma k(t))}{(1-\gamma)^{2} \sigma^{2}(t)}, \\
&=& \frac{\sigma(t)k(t)+\sigma(t)\zeta (t)}{(1-\gamma)^{2} \sigma^{2}(t)} > 0.
\end{eqnarray*}
Therefore, the optimal holding in the risky asset increases with increase in the $\gamma$ of the utility function.
\end{proof}
We now prove a lemma which will be used to prove the next theorem.

\begin{lemma}
$\displaystyle{g(x)=\left(\frac{x^{2}}{2}-x\right)\log x-\frac{x^{2}-4x}{4}},~x>0$ is an increasing function.
\end{lemma}
\begin{proof}
The proof is evident, since,
\[\frac{dg}{dx}=\left(x-1\right) \log x >0.\]
This shows that the function is increasing in the given range. 
\end{proof}

We have considered the approximation of the pay-off function with $x^{\gamma}$, where $\gamma \in \mathbf{R}$. In the following theorem, we show that the best-approximated function (for the objective function, with Limited Liability) has $\gamma \in (0,1)$.

\begin{theorem}
The optimal approximation of the payoff function, with Limited Liability, is given by $x^{\gamma^{*}}$ where $\gamma^{*} \in (0,1) $. 
\end{theorem}
\begin{proof}
The error term in $L^{2}$ norm is given by:
\begin{eqnarray*}
\text{Err} &=& \int\limits_{0}^{B}(\max(F,x)-x^\gamma)^{2}dx, \\
&=& \int\limits_{0}^{F}(F-x^{\gamma})^{2}dx+\int\limits_{F}^{B}(x-x^{\gamma})^{2}dx.
\end{eqnarray*}
We know that $F<1$ (Bankruptcy level) and $B>1$ (Upper bound for net worth in the given timeline) are constants. Now differentiating the error term w.r.t $\gamma$, we get,
\[\frac{d(\text{Err})}{d\gamma}=-2\int\limits_{0}^{F} (F-x^{\gamma})x^{\gamma}\log x dx-2\int\limits_{F}^{B} (x-x^{\gamma})x^{\gamma}\log x dx.\]

\begin{enumerate}[(A)]
\item For the case $\gamma=1$, we get,
\begin{eqnarray*}
\frac{d(\text{Err})}{d\gamma} &=&-2\int\limits_{0}^{F}(F-x)x\log x dx-2 \int\limits_{F}^{B} (x-x)x\log x dx, \\
&=& -2\int\limits_{0}^{F}(F-x)x\log x dx.
\end{eqnarray*}
Now, in order to evaluate the integral, we apply the formula,
$\displaystyle{\int u dv=uv-\int vdu}$ with $u=\log x$ and $\displaystyle{v= \frac{Fx^{2}}{2}-\frac{x^{3}}{3}}$, 
which renders the integral as,
\begin{eqnarray*}
&&\int\limits_{0}^{F} (F-x)x\log x dx, \\
&=&\left[\left(\frac{Fx^{2}}{2}-\frac{x^{3}}{3}\right) \log x\right]_{0}^{F} - \int\limits_{0}^{F}\left(\frac{Fx^{2}}{2}-\frac{x^{3}}{3}\right)dx, \\
&=& \frac{F^{3}}{6} \log F-\frac{5}{36}F^{3} < 0.
\end{eqnarray*}
Since $F<1$, hence $\log x$ is negative, and therefore, 
\[\frac{d(\text{Err})}{d\gamma}=-2\int\limits_{0}^{F} (F-x)x\log x dx > 0.\]
So the optimal $\gamma$ satisfies the condition $\gamma^{*}<1$.
\item For the case of $\gamma=0$, we get, 
\begin{eqnarray*}
\frac{d(\text{Err})}{d\gamma} &=& -2\int\limits_{0}^{F} (F-1) \log x dx-2 \int\limits_{F}^{B} (x-1) \log x dx, \\
&=& -2\left[\left[(F-1)x(\log x-1)\right]_{0}^{F}+ \left[\left(\frac{x^{2}}{2}-x\right)\log x-\frac{x^{2}-4x}{4}\right]_{F}^{B} \right], \\
&=&-2\left[(F-1)F(\log F-1)\right], \\
&-& 2\left[\left[\left(\frac{B^{2}}{2}-B\right)\log B- \frac{B^{2}-4B}{4}\right]-\left[\left(\frac{F^{2}}{2}-F\right)\log F- \frac{F^{2}-4F}{4}\right]\right].
\end{eqnarray*}

Note that, $\displaystyle{g(x)=\left(\frac{x^{2}}{2}-x\right)\log x- \frac{x^{2}-4x}{4}}$, is an increasing function. Hence $g(B)>g(F)$. So the second term in the last equation is positive and the first term $\left[(F-1)F(\log F-1)\right]>0$. Hence, we get,
\[\frac{d(\text{Err})}{d\gamma} < 0.\]
So, if $\gamma$ increases, the error decreases at $\gamma=0$. Hence the optimal $\gamma$ satisfies the condition $\gamma^{*} > 0$.
\end{enumerate}
Therefore, combining both these cases, we conclude that $0 < \gamma^{*} < 1$.
\end{proof}

All these theories presented above lead to the conclusion that the incorporation of Limited Liability leads to inclusion of less risky assets in the portfolio. We state this conclusion, in the form of the following corollary.

\begin{corollary}
The maximum holding in the risky asset continuously changes the value of the objective function. The approximated function for the Limited Liability has lower $\gamma$, as compared to the objective function without Limited Liability. Hence, the above theorem implies that the model with Limited Liability incorporates less risky assets than the objective function without Limited Liability.
\end{corollary}

\section{An Illustrative Example}
\label{Sec_Four_Example}

We solve the optimization problem stated in (\ref{Eq_4.10}) with $U_{1}(X)$ and $U_{2}(X)$, being the utility functions, for the case with Limited Liability and without Limited Liability, respectively. The risky loan has the probability of default $p=0.1$ and the loss given default $\lambda=0.6$. Therefore, we get the upper bound on the investment made on the risky bond as $\delta=0.83$, with the total bound of expected loss being $EL=0.05$.
For this example, we have taken $F=0.75$, the threshold level, below which the bank faces bankruptcy. We choose the bound of the wealth as $B=1.2$, that is, using the bonds, the bank can produce a maximum gain of $20\%$. This is quite natural because of the lower return of the bond. Using these parameter values, the best approximation for the pay-off function with limited liability is given by, 
\[U_{2}(X)=X^{0.1821}.\]
Therefore, we have the two utility functions. Accordingly, we solve the optimization problem and compare the solutions. To solve the optimization problem, the other parameter values are taken as, $\alpha = 0.15$, $b=0.67$, $\theta(t)=0.0075$, $\zeta(t)=0.3$, $T_{1}=1.5$, $T=1$ and the IRB based capital requirement $k=0.2$. Therefore, the capital structure of the firm is $\text{Capital}(t) \geq \max\left(0.04,k\pi(t)\right)$, with $0.04$ being the minimum capital requirement mandated by the supervisor. 

From Figure \ref{FIG:four_portfolio}, we see that when we solve the optimization problem, without Limited Liability, then the portfolio holds as much as the risky loan as possible. On the other hand, with Limited Liability, the portfolio holds less risky loans up to a finite time point, after which it plateaus at the maximum amount of risky loan. This happens because, as time progresses and approaches maturity, the volatility decreases, thereby behaving like a risk free asset. Hence it can be incorporated (near maturity) in the portfolio, without increasing the risk. Thus this result is consistent from the financial point of view.

From Figure \ref{FIG:four_DD} we observe that the $DD$ decreases, as the total risk of the portfolio increases. As the $DD$ decreases, the chances of default increase and the firm becomes less stable. On the other hand, as the time approaches maturity, the volatility is reduced. Consequently, there is an increment in the $DD$, which causes better resiliency of the bank.

\begin{center}
\begin{figure}[!ht]
\includegraphics[width=0.7\linewidth]{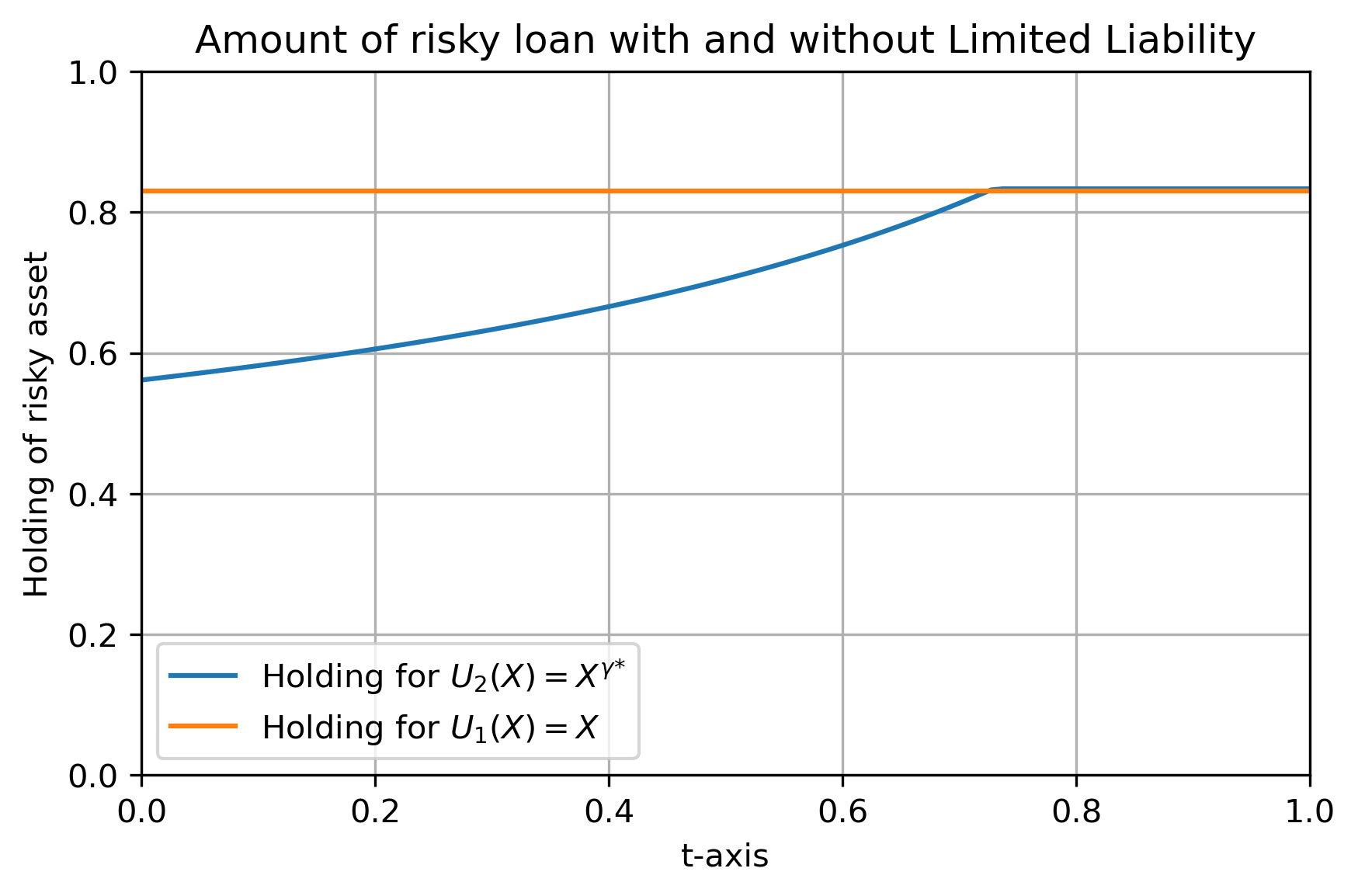}
\caption{Weight of risky loan with and with-out Limited Liability}
\label{FIG:four_portfolio}
\end{figure}
\end{center}
  
\begin{center}
\begin{figure}[!ht]
\includegraphics[width=0.8\linewidth]{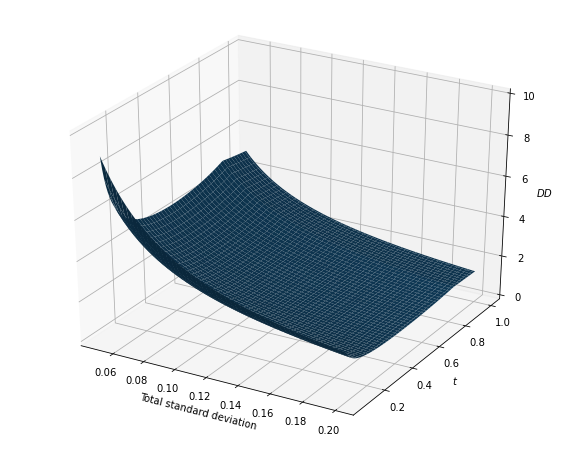}
\caption{Distance of default}
\label{FIG:four_DD}
\end{figure}
\end{center}

\section{Conclusion}
\label{Sec_Four_Conclusion}

Bank's decision problems are usually constructed using a traditional return-risk framework. For a discrete time setup, literature shows that including Limited Liability in the decision model can reduce the risk in the final portfolio. Having said so, the effect of the inclusion of Limited Liability also needs to be considered in the continuous-time step model. To this end, this article focuses on incorporating Limited Liability into the bank's decision problem. 

In this work, we have shown the benefit of combining the Limited Liability in the decision problem. In particular, we have constructed two models to maximize the utility of the terminal wealth, one of which is with Limited Liability and the other is without Limited Liability. The case with Limited Liability is a continuous but non-differentiable function. Therefore, in order to show the effect of the inclusion of Limited Liability, we have approximated it with a $C^{2}$-function, to determine the optimal strategy. Finally, we have proved that the model with Limited Liability leads to less incorporation of risky assets. Banks and other financial firms are subjected to Limited Liability protection and hence the incorporation of Limited Liability reflects a more real-world scenario.

\bibliography{BIBLIO_4}
\bibliographystyle{siam}

\end{document}